\documentclass[a4paper,12pt]{article}
\usepackage{mathtools}
\usepackage{amsfonts}
\usepackage{amsthm}
\usepackage{fullpage}
\newtheorem*{thm}{Theorem}
\newtheorem*{lemma}{Lemma}
\newtheorem*{cor}{Corollary}
\newtheorem*{prop}{Proposition}
\newtheorem*{basicthm}{(Basic) Theorem}
\newtheorem*{basicprop}{Basic Proposition}
\theoremstyle{definition}
\newtheorem*{defi}{Definition}
\newtheorem*{ex}{Example}
\newtheorem*{notat}{Notation}
\newtheorem*{setting}{Setting}
\theoremstyle{remark}
\newtheorem*{note}{Note}

\begin{document}

\title{A more ``complete'' version of the $\pi$-theorem: DRAFT}
\author{Julian Newman,\\ Department of Mathematics,\\ Imperial College\\ London}
\date{2011}
\maketitle

\begin{abstract}
\noindent The traditional Pi-theorem tells us that for any dimensionally invariant relation there exists a full set of independent dimensionless ``Pi groups'' which can be used to nondimensionalise the relation. In this paper, we seek to understand better the structure of dimensionally invariant relations and sets, by giving a complete characterisation of them in terms of independent dimensionless Pi groups. The traditional Pi-theorem only goes part of the way towards achieving such a characterisation. Our characterisation presented here can be viewed as the ``complete Pi-theorem''.
\end{abstract}

\newpage

\section*{0. Introduction}

When we measure physical quantities, the value which we assign is not typically just a real number; rather it is usually a real number multiplied by a ``unit''---that is, a fixed reference value for this particular kind of quantity, a value which is seen as distinct from any ``dimensionless'' real value. All possible values of this ``kind'' of quantity are multiples of this unit (that is, multiples by real, or ``dimensionless'', factors); they are said to have the same ``dimension''. A physical quantity which is not a multiple of this unit is said to be of a different dimension.

$\textrm{}$ \\ It becomes immediately clear from a minimal amount of study of applied mathematics or physics that there should also be a notion of the product of two values, even if neither value is dimensionless (i.e. both are ``dimensional''); and if the two values are dimensional, then the resulting value must be a different ``kind of quantity'' (i.e. of a different dimension) from both of the original quantities, if multiplication is to behave in a similar way as we are used to for standard real numbers. We can refer to the dimension of the result as the product of the dimensions of the two original values. Also, we have a sense of dimensional values being positive, zero or negative; and if a dimensional value is positive, we have a notion of raising this value to a real power.

$\textrm{}$ \\ Typically, we start with a list of ``fundamental dimensions'' such as mass, length and time. We then take these as the basis of a vector space of dimensions, whose addition operator is multiplication of dimensions and whose scalar multiplication operator is raising to scalar (i.e. real) powers. The set of all physical quantities is the set of all values that belong to one of the dimensions in this vector space. The set of all physical quantities that are positive form a vector space under the same operations.

$\textrm{}$ \\ When we perform calculations using some formula that gives the relationship between a list of physical variables, we typically ignore all units in our calculations. This is based on the idea that there is some sense in which the formula can be applied to just the dimensionless numbers which pre-multiply the units, and the answer remains valid regardless of the system of units being used, provided that the system of units is \emph{consistent}.

$\textrm{}$ \\ If we pick any positive value from each of the fundamental dimensions---e.g. kg, ft and min---then we can form a consistent system of units by taking the span of these units. A list of units consisting of one value from each dimension is called a list of fundamentla units. For any system of units that is not equal to a span of fundamental units, there will be a ``clash'' of units; there will be some product of powers of units whose result is \emph{different} in value from the unit for the resulting dimension. For example, length could be measured in cm and time in hr, while the ratio of a length to a time (e.g. a velocity) could be measured in knots. This makes efficient calculations very difficult. From now on, if we talk of a system of units, we will assume that it is consistent.

$\textrm{}$ \\ There is no ``natural'' choice for what system of units to choose. Consequently, we have a notion that all physical relations between physical variables can be characterised by the values which pre-multiply the units, independently of the system of units being used. For example $F = ma$ is a formula that satisfies this property; we can think of it as a relation between dimensionless numbers that pre-multiply the units of force, mass and acceleration, and the result will be valid regardless of the system of units being used. Any physical relation must satisfy this property, otherwise it is not ``treating all systems of units equally'', but instead one can find a physical quantity which is being treated, in at least some sense, as ``special'' among the set of quantities of its dimension.

$\textrm{}$ \\ If some formula does not appear to have this property, then it is because there is a ``hidden dimensional value'' in the formula---but if we include all such dimensional values as physical variables in the formula, then the formula will satisfy this property. As a crude example,

\[ x = \textit{speed of light} \cdot t, \]

\noindent can be regarded as a formula linking the two variables $x$ (distance travelled by a light beam) and $t$ (time for which the beam has been travelling). Suppose that we have a pair of values for $x$ and $t$ which satisfy the formula; if we then change the units of these values while keeping the numbers that pre-multiply the units the same, the resulting values will no longer satisfy the formula. However, this is because there is a constant dimensional value, the speed of light, that the formula ``treats as special''. But if we regard the formula as a formula linking the \emph{three} variables $x$, $t$ and the speed of light $c$ (even though, as far as the \emph{physics} is concerned $c$ is a constant), then the formula does remain invariant under a switch in the units for $x$, $t$ and $c$ (provided the units are consistent).

$\textrm{}$ \\ If a relation between physical variables satisfies this property that, in intuitive terms, all dimensional quantities involved are explicitly present as physical variables, shall be called \emph{dimensionally invariant}. (This term is common among those working in ``measurement theory'', e.g. [1], and related areas, e.g. [2]. Sometimes such relations are just called a ``physical relations'', as in [3]). For example, ``$F=ma$'' is a dimensionally invariant relation between $F$, $m$ and $a$. The $\pi$-theorem tells us that any dimensionally invariant relation between positive physical variables can be ``nondimensionalised'' using products of powers of the original variables, or ``$\pi$ groups''; it furthermore tells us (at most) how many of these $\pi$ groups are needed. (For physical variables that can also take non-positive values, the sgn function might also be needed.)

$\textrm{}$ \\ In this paper, we will show that the $\pi$-theorem is just a weak version of a one-way implication of an ``if and only if'' fact. We will prove the full fact, the ``full $\pi$-theorem'', which gives a clearer insight into the structure of dimensionally invariant relations.

$\textrm{}$ \\ But first we must introduce our terminology and notation, and the important basic facts of dimensional analysis.

\section*{1. Introductory concepts}

A note on terminology: In the context of vector spaces, we shall consider ``linearly independent (ordered) lists'' of vectors, and the ``bases'' shall be ordered lists, rather than unordered sets.

\textrm{} \\ \emph{Just as} $\Box$ \emph{is used to denote the end of a proof, we shall use} $\triangle$ \emph{to denote the end of a Note, Remark, Definition, statement of a Proposition, etc.}

\begin{setting}
Let $X$ be a vector space over a field $\mathbb{F}$. Let $V$ be a subspace of $X$, and let $W$ be the quotient space $X/V$.

$\textrm{}$ \\ Denote the sum of two values $x_1$ and $x_2$ in $X$ as $x_1 x_2$. For a value $\lambda \in \mathbb{F}$ and $x \in X$, denote the product of $\lambda$ with $x$ as $x^{\lambda}$. Denote the additive inverse of an element $x \in X$ as $x^{-1}$. All of these notations shall also apply for operations on elements of $W$. Denote the zero vector of $X$ as $1$. 

$\textrm{}$ \\ Denote the additive identity of $\mathbb{F}$ as $0_{\mathbb{F}}$.

$\textrm{}$ \\ Let $q:X \to W$ be the projection sending an element $x \in X$ onto its coset $Vx \in W$. \hfill $\triangle$
\end{setting}

\noindent In our context, $X$ is the set of all positive physical quantities, $\mathbb{F}$ is $\mathbb{R}$ or $\mathbb{Q}$, $V$ is the set of positive dimensionless numbers, and $W$ is the set of dimensions, where a dimension is viewed as a complete set of positive quantities that are multiples of each other (for example, the dimension ``length'' is the set of all positive length values).

\begin{defi}
A list of values $x_1, \ldots, x_n$ in $X$ is called \emph{consistent} if $q$ is injective on $\textrm{span}(x_1, \ldots, x_n)$---i.e. if
\[ q|_{\textrm{span}(x_1, \ldots, x_n)}: \textrm{span}(x_1, \ldots, x_n) \to \textrm{span} \left( q(x_1), \ldots, q(x_n) \right) \]
is a vector space isomorphism. \hfill $\triangle$
\end{defi}

\noindent Thus, a list of units of measurement is consistent if it gives rise to no ``clashes''. A clash is when one product of powers of the units gives one value, while another product of powers of the units gives another value \emph{of the same dimension}.

\begin{prop}[alternative definition of `consistent']
A list of values $x_1, \ldots, x_n$ is consistent if and only if there exist values $y_1, \ldots, y_m$ such that
\[ q(y_1), \ldots, q(y_m) \textrm{ are linearly independent, and} \]
\[ x_1, \ldots, x_n \in \textrm{span}(y_1, \ldots, y_n). \]
\hfill $\triangle$
\end{prop}

\noindent Thus, a list of units is consistent if and only if it is contained in the span of a list of values which could be treated as ``fundamental units'' (that is, a list of values whose dimensions are linearly indepedent). The values $\textrm{V}$, $\textrm{A}$, $\Omega$, $\textrm{s}$ and $\textrm{F}$ form a consistent list, because they are in $\textrm{span}(V,A,s)$, and the dimensions $\textrm{[V]}$, $\textrm{[A]}$ and $\textrm{[s]}$ are linearly independent. Indeed, much of basic electronics could be done using $\textrm{V}$, $\textrm{A}$ and $\textrm{s}$ as the fundamental units, rather than the standard $\textrm{kg}$, $\textrm{m}$, $\textrm{A}$ and $\textrm{s}$.

\begin{proof}
``if'': Let $y_1, \ldots, y_m$ be values such that $q(y_1), \ldots, q(y_m)$ are linearly indepedent. First we will show that $y_1, \ldots, y_m$ are consistent: Take two values $s$ and $t$ in $\textrm{span}(y_1, \ldots, y_m)$, such that $q(s)=q(t)$. Express $s$ and $t$ as
\begin{align*}
s &= y_1^{\lambda_1} \ldots y_m^{\lambda_m} \textrm{ and} \\
t &= y_1^{\mu_1} \ldots y_m^{\mu_m}.
\end{align*}
Then
\[ q(s) = \underline{ q(y_1)^{\lambda_1} \ldots q(y_m)^{\lambda_m} = q(y_1)^{\mu_1} \ldots q(y_m)^{\mu_m} } = q(t) \]
and so
\[ \lambda_1 = \mu_1, \lambda_2 = \mu_2, \ \ldots\ , \ \lambda_m = \mu_m. \]
Hence $s = t$.

$\textrm{}$ \\ Hence $y_1, \ldots, y_m$ are consistent. Now take any list of values $x_1, \ldots, x_n \in \textrm{span}(y_1, \ldots, y_n)$. Since $q$ is injective on $\textrm{span}(y_1, \ldots, y_m)$, it is in particular injective on the subset

\[\textrm{span}(x_1, \ldots, x_n) \subset \textrm{span}(y_1, \ldots, y_m).\]

\noindent So $x_1, \ldots, x_n$ are consistent.

$\textrm{}$ \\ ``only if'': Take any consistent list $x_1, \ldots, x_n$. Let $y_1, \ldots, y_m$ be a basis of $\textrm{span}(x_1, \ldots, x_n)$. Then, since $q|_{\textrm{span}(x_1, \ldots, x_n)}$ is an isomorphism, it follows that $q(y_1), \ldots, q(y_m)$ form a basis of $\textrm{span}(q(x_1),\ldots,q(x_n))$ (and hence in particular are linearly independent).
\end{proof}

\begin{note}
It immediately follows from the above proof that in the proposition, we could replace the line
\[ q(y_1), \ldots, q(y_m) \textrm{ are linearly independent} \]
with the stronger statement
\[ q(y_1), \ldots, q(y_m) \textrm{ form a basis of span} \left( q(x_1), \ldots, q(x_n) \right). \]
\hfill $\triangle$
\end{note}

\begin{defi}
For any positive integer $k$, a $k$\emph{-input fixed-coefficient-linear-combination function} (or `FCLCF') is a function
\[ p: X^k \to X \]
for which there exist values $\lambda_1, \ldots, \lambda_k \in \mathbb{F}$ such that
\[ p(x_1, \ldots, x_k) = x_1^{\lambda_1} \ldots x_k^{\lambda_k} \]
for all $x_1, \ldots, x_k \in X$.
We also define the 0-input FCLCF as being the map that sends the empty tuple $()$ onto the zero-vector 1.
\hfill $\triangle$
\end{defi}

\begin{note}
By completely elementary linear algebra, for any FCLCF the coefficients $\lambda_1, \ldots, \lambda_k \in \mathbb{F}$ are unique.
\hfill $\triangle$
\end{note}

\begin{defi}
Let $p$ be a $k$-input FCLCF, with coefficients $\lambda_1, \ldots, \lambda_k$. Then, for any list of elements $w_1,\ldots,w_k$ of $W$, we define
\[ p(w_1,\ldots,w_k) := w_1^{\lambda_1} \ldots w_k^{\lambda_k}. \]
Note that if $x_i \in w_i$ for each $i$ from $1$ to $k$, then
\[ p(x_1, \ldots, x_k) \in p(w_1,\ldots,w_k). \]
\hfill $\triangle$
\end{defi}

\begin{note}
Suppose that $p$ is a $k$-input FCLCF. Then for any $x_1, \ldots, x_k, y_1, \ldots, y_k \in X$, and $\lambda \in \mathbb{F}$,
\begin{align*}
p(x_1 y_1, \ldots, x_k y_k) &= p(x_1, \ldots, x_k) p(y_1, \ldots, y_k), \textrm{ and} \\ 
p(x_1^{\lambda}, \ldots, x_k^{\lambda}) &= \left( p(x_1, \ldots, x_k) \right)^{\lambda}.
\end{align*}
\hfill $\triangle$
\end{note}

\begin{basicprop}
Let $w$ be an element of $W$, and tet $s \in w$.
Then for any $x \in w$ there exists a unique $a \in V$ such that $x = as$.
\hfill $\triangle$
\end{basicprop}

\begin{proof}
Existence: Take $a := x s^{-1}$. Uniqueness: If $x = as$, then $a = x s^{-1}$.
\end{proof}

\begin{defi}
Take any $w_1, \ldots, w_n \in W$. A We shall define a ``relation on $(w_1, \ldots, w_n)$'' as a function
\[ f: w_1 \times \ldots \times w_n \to \{ \textrm{TRUE}, \textrm{FALSE} \} \]
that assigns a binary value ``$\textrm{TRUE}$'' or ``$\textrm{FALSE}$'' to each element of the Cartesian product $w_1 \times \ldots \times w_n$.

$\textrm{}$ \\ An ``$n$-input dimensionless relation'' is a function
\[ f: V^n \to \{ \textrm{TRUE}, \textrm{FALSE} \}. \]
\hfill $\triangle$
\end{defi}

\noindent An example of a relation might be:
\[ f: \mathsf{M} \times \mathsf{M} \mathsf{T}{}^{-2} \times \mathsf{T} \to \{ \textrm{TRUE}, \textrm{FALSE} \} \]
\[ f(m, k, t) = \textrm{TRUE} \iff \frac{t}{2 \pi} \sqrt{\frac{k}{m}} \in \mathbb{Z}_{>0}. \]

\begin{defi}
A relation $f$ on $(w_1, \ldots, w_n)$ is called \emph{dimensionally invariant} if
for all $a_1, \ldots, a_n \in V$ and any pair of consistent lists $(s_1, \ldots, s_n)$ and $(t_1, \ldots, t_n)$ in $w_1 \times \ldots \times w_n$,
\[ f(a_1 s_1, \ldots, a_n s_n) = f(a_1 t_1, \ldots, a_n t_n). \]
\hfill $\triangle$
\end{defi}

\begin{prop}
Let $g$ be an $r$-input dimensionless relation. Let $w_1, \ldots, w_n$ be elements of $W$. Let $\varphi_1, \ldots, \varphi_r$ be $n$-input FCLCFs such that
\[ \varphi_i[w_1 \times \ldots \times w_n] \subset V \textrm{ for all } i \textrm{ from } 1 \textrm{ to } r. \]
Then the relation $f$ on $(w_1, \ldots, w_n)$ defined by
\[ f(x_1, \ldots, x_n) = g \left( \varphi_1(x_1, \ldots, x_n), \ldots, \varphi_r(x_1, \ldots, x_n) \right) \]
is dimensionally invariant. \hfill $\triangle$
\end{prop}

\noindent For example, the above relation
\[ f(m, k, t) = \textrm{TRUE} \iff \frac{t}{2 \pi} \sqrt{\frac{k}{m}} \in \mathbb{Z}_{>0} \]
is dimensionally invariant: $r=1$, $n=3$, $g(x)=\textrm{TRUE} \textrm{ iff } \frac{x}{2 \pi} \in \mathbb{Z}_{>0}$, and $\varphi_1(m,k,t)=m^{\frac{1}{2}}k^{-\frac{1}{2}}t$.

\begin{proof}
Take any $a_1, \ldots, a_n \in V$ and any pair of consistent lists $(s_1, \ldots, s_n)$ and $(t_1, \ldots, t_n)$ in $w_1 \times \ldots \times w_n$. For each $i$ from $1$ to $r$,
\[ \varphi_i(s_1, \ldots, s_n) \in V. \]
Also,
\[ s_1^{0_{\mathbb{F}}} \ldots s_n^{0_{\mathbb{F}}} = 1 \in V. \]
Hence, since $(s_1, \ldots, s_n)$ is consistent,
\[ \varphi_i(s_1, \ldots, s_n) = 1. \]
Similarly,
\[ \varphi_i(t_1, \ldots, t_n) = 1. \]
Hence,
\begin{align*}
  & g \left( \varphi_1(a_1 s_1, \ldots, a_n s_n), \ldots, \varphi_r(a_1 s_1, \ldots, a_n s_n) \right) \\
= & g \left( \varphi_1(a_1, \ldots, a_n), \ldots, \varphi_r(a_1, \ldots, a_n) \right) \\ 
= & g \left( \varphi_1(a_1 t_1, \ldots, a_n t_n), \ldots, \varphi_r(a_1 t_1, \ldots, a_n t_n) \right).
\end{align*}
\end{proof}

\begin{defi}
Let $L_n$ be the vector space of all n-input FCLCFs under the operations:
\begin{align*}
\textrm{(addition) } & p_1p_2:(x_1, \ldots, x_n) \mapsto p_1(x_1, \ldots, x_n)p_2(x_1, \ldots, x_n) \\ 
\textrm{(multiplication) } & p^{\lambda}:(x_1, \ldots, x_n) \mapsto p(x_1, \ldots, x_n)^{\lambda}.
\end{align*}
\hfill $\triangle$
\end{defi}

\begin{defi}
Take any elements $w_1, \ldots, w_n$ of $W$. We define $K[w_1,\ldots,w_n] \subset L_n$ as the set of $n$-input FCLCFs $p$ such that
\[ p[w_1 \times \ldots \times w_n] \subset V. \]
(In other words, $K[w_1,\ldots,w_n]$ is the set of $n$-input FCLCFs $p$ such that
\[ p(w_1, \ldots, w_n) = \mathbf{0}_W = V.) \]
\hfill $\triangle$
\end{defi}

\begin{thm}
Let $w_1, \ldots, w_n$ be elements of $W$, and let $m$ be the vector space dimension of $\textrm{span}(w_1, \ldots, w_n)$. Then $K[w_1,\ldots,w_n]$ is an $(n-m)$-dimensional subspace of $L_n$.
\hfill $\triangle$
\end{thm}

\begin{proof}
For any $w_1, \ldots, w_n \in W$, define
\[ T[w_1, \ldots, w_n]: L_n \to W \]
\[ T[w_1, \ldots, w_n](p) = p(w_1, \ldots, w_n) \textrm{ for all } p \in L_n. \]
Clearly $T[w_1, \ldots, w_n]$ is linear. The image $T[w_1, \ldots, w_n][L_n]$ is $\textrm{span}(w_1, \ldots, w_n)$. The kernel is $K[w_1,\ldots,w_n]$.
Hence, by the rank-nullity theorem,
\[ \textrm{dim}K[w_1, \ldots, w_n] = n-m. \]
\end{proof}

\begin{note}
Let $w_1, \ldots, w_n$ be elements of $W$, and let $m$ be the vector space dimension of $\textrm{span}(w_1, \ldots, w_n)$. Then there exist $n$-input FCLCFs $\pi_1, \ldots, \pi_{n-m}$ such that: \\

for any relation on $(w_1, \ldots, w_n)$ of the form
\[ f(x_1, \ldots, x_n) = g \left( \varphi_1(x_1, \ldots, x_n), \ldots, \varphi_r(x_1, \ldots, x_n) \right), \]

we can find an $(n-m)$-input dimensionless relation $g'$ such that
\[g \left( \varphi_1(x_1,\ldots,x_n), \ldots, \varphi_r(x_1,\ldots,x_n) \right) = g' \left( \pi_1(x_1,\ldots,x_n), \ldots, \pi_{n-m}(x_1,\ldots,x_n) \right). \]
The (more elementary) proof of this fact is as follows:
Let $\pi_1, \ldots, \pi_{n-m}$ be a basis of $K[w_1, \ldots, w_n]$.
Now take any $r$-input dimensionless relation $g$ and any FCLCFs $\varphi_1, \ldots, \varphi_r \in K[w_1, \ldots, w_n]$.
For each $i$ from $1$ to $r$, let $\psi_i$ be the $r$-input fixed-coefficient-linear-combination function such that
\[ \phi_i(.) = \psi_i \left( \pi_1(.), \ldots, \pi_{n-m}(.) \right). \]
Define an $(n-m)$-input dimensionless relation $g'$ by
\[ g'(v_1, \ldots, v_{n-m}) = g \left( \psi_1(v_1, \ldots, v_{n-m}), \ldots, \psi_r(v_1, \ldots, v_{n-m}) \right). \]
Then
\[g \left( \varphi_1(x_1,\ldots,x_n), \ldots, \varphi_r(x_1,\ldots,x_n) \right) = g' \left( \pi_1(x_1,\ldots,x_n), \ldots, \pi_{n-m}(x_1,\ldots,x_n) \right)\! . \textit{ Q\!.E.D.} \]

$\textrm{}$ \\ The above statement should not be confused with the $\pi$-theorem itself. One of the uses of nondimensionalisation is to reduce the number of parameters involved. If a relation has already been expressed in a nondimensionalised form, then the above statement may provide a yet more ``efficient'' nondimensionalisation. Nonetheless, we have given it the status of a ``Note'' rather than a ``Theorem'', because it is merely an \emph{example} of the fact that $\pi$-theorem is true. (In fact, in and of itself, the above statement has nothing to do with dimensional invariance.) A dimensionally invariant relation may happen to be expressed \emph{not} in a nondimensionalised form, but the $\pi$-theorem will tell us that it can still be nondimensionalised, and with the same ``efficiency'' as in the above.
\hfill $\triangle$
\end{note}

\begin{defi}
For each $i$ from $1$ to $n$, let the FCLCF $P_i \in L_n$ be the projection onto the $i$-th input,
\[ P_i : (x_1,\ldots,x_n) \mapsto x_i. \]
(If we were to be pedantic, we should specify, in the notation for these projections, the number of inputs---e.g. by writing $P_i^{(n)}$ for the $i$-th $n$-input projection. But just as we typically omit the $^{(n)}$ in the symbol $I^{(n)}$ for the $n \times n$ identity matrix, so we shall usually omit the $^{(n)}$ here.)
\hfill $\triangle$
\end{defi}

\begin{prop}
$L_n$ is $n$-dimensional, with the projections $P_1,\ldots,P_n$ forming a basis.
\hfill $\triangle$
\end{prop}

\begin{proof}
Easy.
\end{proof}

\begin{defi}
Let $x_1, \ldots, x_n$ be a linearly independent list of elements of $X$, and let $y \in \textrm{span}(x_1, \ldots, x_n)$. Then we define
\[ p[x_1, \ldots, x_n; y] \]
as the unique $n$-input FCLCF such that
\[ p[x_1, \ldots, x_n; y](x_1, \ldots, x_n) = y. \]
We will also use the same notation for elements of $W$ rather than $X$.
\hfill $\triangle$
\end{defi}

\begin{lemma}
Let $w_1, \ldots, w_n$ be elements of $W$. Let $\beta_1, \ldots, \beta_m$ be a basis of $\textrm{span}(w_1, \ldots, w_n)$.
Take any consistent list $(s_1, \ldots, s_n) \in w_1 \times \ldots \times w_n$. Then there exist unique
\[ x_1, \ldots, x_m \textrm{ respectively in } \beta_1, \ldots, \beta_m \]
such that
\[ p[\beta_1, \ldots, \beta_m; w_i](x_1, \ldots, x_m) = s_i \]
for each $i$ from $1$ to $n$.
\hfill $\triangle$
\end{lemma}

\noindent Note that here, a basis of $\textrm{span}(w_1, \ldots, w_n)$ can be selected \emph{before} the consistent list $(s_1, \ldots, s_n)$ is given.

\begin{proof}
Existence: Take any consistent list $(s_1, \ldots, s_n) \in w_1 \times \ldots \times w_n$. For each $j$ from $1$ to $m$,
\[ \beta_j \in \textrm{span}(w_1, \ldots, w_n) = \textrm{span} \left( q(s_1), \ldots, q(s_n) \right) \]
and so we can define $x_j$ by
\[ x_j = \left( q|_{\textrm{span}(s_1, \ldots, s_n)} \right)^{-1}(\beta_j). \]
Clearly, since $q|_{\textrm{span}(s_1, \ldots, s_n)}$ is an isomorphism,
\[ p[\beta_1, \ldots, \beta_m; w_i](x_1, \ldots, x_m) = s_i \]
for each $i$ from $1$ to $n$.

$\textrm{}$ \\ Uniqueness: Take any $x_1, \ldots, x_m \textrm{ respectively in } \beta_1, \ldots, \beta_m$ such that
\[ p[\beta_1, \ldots, \beta_m; w_i](x_1, \ldots, x_m) = s_i \]
for each $i$ from $1$ to $n$. (We will determine explicitly the values of $x_1, \ldots, x_m$, thus proving uniqueness.)

$\textrm{}$ \\ Since $q(x_1),\ldots,q(x_m)$ are linearly independent, $x_1,\ldots,x_m$ are themselves linearly independent, and so $\textrm{span}(x_1,\ldots,x_m)$ has dimension $m$. Since $q|_{\textrm{span}(s_1, \ldots, s_n)}$ is an isomorphism, $\textrm{span}(s_1, \ldots, s_n)$ also has dimension $m$.

$\textrm{}$ \\ Now $\textrm{span}(s_1, \ldots, s_n) \subset \textrm{span}(x_1,\ldots,x_m)$, since $s_i \in \textrm{span}(x_1,\ldots,x_m)$ for each $i$. But since $\textrm{span}(s_1, \ldots, s_n)$ and $\textrm{span}(x_1,\ldots,x_m)$ both have the same dimension, it also follows that $\textrm{span}(x_1,\ldots,x_n) \subset \textrm{span}(s_1,\ldots,s_m)$.

$\textrm{}$ \\ Thus, for each $j$ from $1$ to $m$,
\[ x_j \in \textrm{span}(s_1, \ldots, s_n) \: \textrm{and} \: x_j \in \beta_j \in \textrm{span} \left( q(s_1), \ldots, q(s_n) \right). \]
Thus, since $q$ is injective on $\textrm{span}(s_1, \ldots, s_n)$, we must have that
\[ x_j = \left( q|_{\textrm{span}(s_1, \ldots, s_n)} \right)^{-1}(\beta_j). \]
\end{proof}

\begin{basicthm}
Let $w_1, \ldots, w_n$ be elements of $W$. Let $b_1, \ldots, b_m$ be linearly independent elements of $W$ whose span contains $\textrm{span}(w_1, \ldots, w_n)$.
Then for all consistent lists $(s_1, \ldots, s_n) \in w_1 \times \ldots \times w_n$, there exist (not necessarily unique)
\[ x_1, \ldots, x_m \textrm{ respectively in } b_1, \ldots, b_m \]
such that
\[ p[b_1, \ldots, b_m; w_i](x_1, \ldots, x_m) = s_i \]
for each $i$ from $1$ to $n$.
\hfill $\triangle$
\end{basicthm}

\begin{proof}
Let $\beta_1, \ldots, \beta_k$ be a basis of $\textrm{span}(w_1, \ldots, w_n)$. Let $\beta_{k+1},\ldots,\beta_m$ be elements of $X$ such that $\beta_1,\ldots,\beta_k,\beta_{k+1},\ldots,\beta_m$ form a basis of $\textrm{span}(b_1, \ldots, b_m)$.

\textrm{} \\ Let $y_1, \ldots, y_k \textrm{ respectively in } \beta_1,\ldots,\beta_k$ be such that
\[ p[\beta_1, \ldots, \beta_k; w_i](y_1, \ldots, y_k) = s_i \]
for each $i$ from $1$ to $n$. Let $y_{k+1},\ldots,y_m$ be any elements respectively of $\beta_{k+1},\ldots,\beta_m$.

$\textrm{}$ \\ The list $(y_1,\ldots,y_k,y_{k+1},\ldots,y_m)$ is consistent. So let $x_1,\ldots,x_m \textrm{ respectively in } b_1,\ldots,b_m$ be such that
\[ p[b_1, \ldots, b_m; q(y_i)](x_1, \ldots, x_m) = y_i \]
for each $i$ from $1$ to $m$ --- and so in particular for each $i$ from $1$ to $k$.

$\textrm{}$ \\ We can express every member of $(s_1, \ldots, s_n)$ as a linear combination of (i.e. as the result of an FCLCF applied to) $(y_1,\ldots,y_k)$, and can express each member of $(y_1,\ldots,y_k)$ as a linear combination of $(x_1,\ldots,x_m)$. It follows that we can express every member of $(s_1, \ldots, s_n)$ as a linear combination of $(x_1,\ldots,x_m)$; say
\[ p_i(x_1,\ldots,x_m) = s_i \]
for each $i$ from $1$ to $n$. Obviously, each $p_i$ then satisfies the property that
\[ p_i(b_1,\ldots,b_m) = w_i, \]
i.e.
\[ p_i = p[b_1, \ldots, b_m; w_i]. \]
So we are done.
\end{proof}

\noindent Our ``alternative definition'' of a consistent list of units (earlier on) was that there exists some choice of fundamental dimensions such that every unit in the list is in the span of some choice of fundamental units. Now, we are saying that even if we choose from the outset to work with some particular list of fundamental dimensions, a list of units is consistent if and only if it can be expressed as the span of some list of fundamental units.

\begin{prop}[alternative definitions of dimensionally invariant relations]
Let $w_1, \ldots, w_n$ be elements of $W$, and $f$ a relation on $(w_1, \ldots, w_n)$.
Then the following are equivalent:
\begin{enumerate}
\item $f$ is dimensionally invariant
\item there exists a linearly independent list $b_1, \ldots, b_m$ of elements of $W$ whose span contains $\textrm{span}(w_1, \ldots, w_n)$ such that

\textrm{ \indent } \textrm{ \indent } \textrm{ \indent } for any $a_1, \ldots, a_n \in V$ and $(x_1, \ldots, x_m),(y_1, \ldots, y_m) \in b_1 \times \ldots \times b_m$,
\begin{align*}
   & f \left( a_1 p[b_1, \ldots, b_m; w_1](x_1, \ldots, x_m), \ldots, a_n p[b_1, \ldots, b_m; w_n](x_1, \ldots, x_m) \right) \\ 
 = & f \left( a_1 p[b_1, \ldots, b_m; w_1](y_1, \ldots, y_m), \ldots, a_n p[b_1, \ldots, b_m; w_n](y_1, \ldots, y_m) \right)
\end{align*}
\item for \emph{every} linearly independent list $b_1, \ldots, b_m$ of elements of $W$ whose span contains $\textrm{span}(w_1, \ldots, w_n)$,

\textrm{ \indent } \textrm{ \indent } \textrm{ \indent } for any $a_1, \ldots, a_n \in V$ and $(x_1, \ldots, x_m)$, $(y_1, \ldots, y_m) \in b_1 \times \ldots \times b_m$,
\begin{align*}
   & f \left( a_1 p[b_1, \ldots, b_m; w_1](x_1, \ldots, x_m), \ldots, a_n p[b_1, \ldots, b_m; w_n](x_1, \ldots, x_m) \right) \\ 
 = & f \left( a_1 p[b_1, \ldots, b_m; w_1](y_1, \ldots, y_m), \ldots, a_n p[b_1, \ldots, b_m; w_n](y_1, \ldots, y_m) \right).
\end{align*}
\end{enumerate}
\hfill $\triangle$
\end{prop}

\begin{proof}
Easy, using the above theorem.
\end{proof}

\noindent The above proposition highlights what the concept of dimensional invariance really is: no matter how we rescale the fundamental dimensions, the result remains the same.

$\textrm{}$ \\ We now give the $\pi$-theorem. The same theorem can be given using other frameworks and definitions than the ones developed above (see, for example, [4]), but any proof that is not over-complicated will be essentially the same in its key points as the one given here.

\begin{thm}[$\pi$-Theorem]
Let $w_1, \ldots, w_n$ be elements of $W$, and $f$ a dimensionally invariant relation on $(w_1, \ldots, w_n)$. Let $m$ be the dimension of $\textrm{span}(w_1, \ldots, w_n)$. Then there exists an $(n-m)$-input dimensionless relation $g$ and $n$-input FCLCFs $\pi_1, \ldots, \pi_{n-m}$ such that
\[ \pi_i[w_1 \times \ldots \times w_n] \subset V \textrm{ for all } i \textrm{ from 1 to } n-m \]
and, for all $(x_1, \ldots, x_n) \in w_1 \times \ldots \times w_n$,
\[ f(x_1, \ldots, x_n) = g \left( \pi_1(x_1, \ldots, x_n), \ldots, \pi_{n-m}(x_1, \ldots, x_n) \right). \]
\hfill $\triangle$
\end{thm}

\begin{proof}
Take any $(x_1, \ldots, x_n) \in w_1 \times \ldots \times w_n$.

$\textrm{}$ \\ Let $ \{ r_1, \ldots, r_m \} \subset \{ 1, \ldots, n \}$ be such that $w_{r_1}, \ldots, w_{r_m}$ is a basis of $\textrm{span}(w_1, \ldots, w_n)$. Then by the ``alternative definition of consistent'', the list
\[ \left( p[w_{r_1}, \ldots, w_{r_m};w_1](x_{r_1}, \ldots, x_{r_m}), \ldots, p[w_{r_1}, \ldots, w_{r_m};w_n](x_{r_1}, \ldots, x_{r_m}) \right) \]
is consistent. Note that it is also in $w_1 \times \ldots \times w_n$. For the sake of tidiness, let
\[ p_i = p[w_{r_1}, \ldots, w_{r_m};w_i] \]
for each $i$ from $1$ to $n$.
So
\[ \left( p_1 (x_{r_1}, \ldots, x_{r_m}), \ldots, p_n (x_{r_1}, \ldots, x_{r_m}) \right) \]
is consistent.
Let $f'$ be the $n$-input dimensionless relation given by
\[ f'(a_1, \ldots, a_n) = f(a_1 s_1, \ldots, a_n s_n) \textrm{ for all } a_1, \ldots a_n \in V \]
where $(s_1, \ldots, s_n)$ can be any consistent list in $w_1 \times \ldots \times w_n$; by definition of $f$ being dimensionally invariant, $f'$ is independent of the choice of $(s_1, \ldots, s_n)$.
Then,
\begin{align*}
& f(x_1, \ldots, x_n) \\
= & f \left( \left( x_1 p_1(x_{r_1}, \ldots, x_{r_m})^{-1} \right) p_1(x_{r_1}, \ldots, x_{r_m}), \ldots, \left( x_n p_n(x_{r_1}, \ldots, x_{r_m})^{-1} \right) p_n(x_{r_1}, \ldots, x_{r_m}) \right) \\ 
= & f' \left( x_1 p_1(x_{r_1}, \ldots, x_{r_m})^{-1}, \ldots, x_n p_n(x_{r_1}, \ldots, x_{r_m})^{-1} \right).
\end{align*}
We now have an expression that is nearly in the desired form; only, it has $n$ FCLCFs as arguments of an $n$-input dimensionless relation $f'$. Note, however, that for each $i$ from $1$ to $m$, the function
\[ p_{r_i} = p[w_{r_1}, \ldots, w_{r_m};r_i] \]
is just the function that outputs its $i^{\textrm{th}}$ input. In particular,
\[ p_{r_i}(x_{r_1}, \ldots, x_{r_m}) = x_{r_i}. \]
Hence, in the above expression of $f$ in terms of $f'$, $m$ of the $n$ arguments of $f'$ are just the constant value 1.

$\textrm{}$ \\ So we can can label the non-trivial arguments of $f'$ as $\pi_1(x_1, \ldots, x_n), \ldots, \pi_{n-m}(x_1, \ldots, x_n)$, and define a ``condensed'' $n-m$-input dimensionless relation $g$, such that
\[ f(x_1, \ldots, x_n) = g \left( \pi_1(x_1, \ldots, x_n), \ldots, \pi_{n-m}(x_1, \ldots, x_n) \right). \]
\end{proof}

\begin{note}
Observe that the FCLCFs $\pi_1, \ldots, \pi_{n-m}$ in the above proof are linearly independent, and so form a basis of $K[w_1, \ldots, w_n]$. Some people require in the statement of the $\pi$-theorem that $\pi_1, \ldots, \pi_{n-m}$ are ``independent'' in some appropriate sense (see, for example, [5]).

$\textrm{}$ \\ Also note that, whereas the theorem is typically formulated in such a way as to suggest that the FCLCFs are selected \emph{after} the dimensionally invariant relation is given, in actual fact the proof makes clear that the FCLCFs only depend on the dimension $w_1,\ldots,w_n$ and can be chosen \emph{before} a dimensionally invariant relation is given.
\hfill $\triangle$
\end{note}

\noindent This completes Section 1. We will now go on to consider ``dimensionally invariant sets''. A dimensionally invariant subset of a Cartesian product of dimensions $w_1 \times \ldots \times w_n$ is a subset which remains the same after a switching the system of units. Another way of understanding a dimensionally invariant subset is that there exists a dimensionally invariant relation which holds true precisely on this subset.

\section*{2. Basic properties of dimensionally invariant sets}

\begin{defi}
Let $w_1,\ldots,w_n$ be elements of $W$. Let $S$ be a subset of $w_1 \times \ldots \times w_n$. We say that $S$ is a \emph{dimensionally invariant subset} of $w_1 \times \ldots \times w_n$ if for any pair of consistent lists $(s_1,\ldots,s_n)$ and $(t_1,\ldots,t_n)$ in $w_1 \times \ldots \times w_n$,
\[ \{ (v_1,\ldots,v_n) | (v_1s_1,\ldots,v_ns_n) \in S \} = \{ (v_1,\ldots,v_n) | (v_1t_1,\ldots,v_nt_n) \in S \}. \]
\hfill $\triangle$
\end{defi}

\begin{note}
Let $S$ be a subset of $w_1 \times \ldots \times w_n$. Then for \emph{any} list $(s_1,\ldots,s_n) \in w_1 \times \ldots \times w_n$,
\[ \{ (v_1,\ldots,v_n) | (v_1s_1,\ldots,v_ns_n) \in S \} = \{ (x_1 s_1^{-1}, \ldots, x_n s_n^{-1}) | (x_1,\ldots,x_n) \in S \}. \]
\hfill $\triangle$
\end{note}

\noindent Intuitively, a subset of a Cartesian product of dimensions is dimensionally invariant if the set of dimensionless values that is obtained by ``ignoring the units'' is independent of the consistent list of units being used (or rather ignored).

\begin{defi}
Let $w_1,\ldots,w_n$ be elements of $W$. Let $S$ be a dimensionally invariant subset of $w_1 \times \ldots \times w_n$. Define
\[ \overline{S} = \{ (v_1,\ldots,v_n) | (v_1s_1,\ldots,v_ns_n) \in S \} \]
where $(s_1,\ldots,s_n)$ is any consistent list in $w_1 \times \ldots \times w_n$.
\hfill $\triangle$
\end{defi}

\begin{ex}
Let $S = w_1 \times \ldots \times w_n$. Clearly $S$ is dimensionally invariant and
\[ \overline{S} = V^n. \]
\hfill $\triangle$
\end{ex}

\begin{prop}[alternative definitions of `dimensionally invariant subset']
Let $w_1,\ldots,w_n$ be elements of $W$. Let $S$ be a subset of $w_1 \times \ldots \times w_n$. Then the following are equivalent:
\begin{enumerate}
\item
$S$ is dimensionally invariant
\item
the indicator function
\[ I_S: w_1 \times \ldots \times w_n \to \{ \textrm{TRUE}, \textrm{FALSE} \} \]
\[ I_S(\mathbf{x}) = \textrm{TRUE} \iff \mathbf{x} \in S \]
is a dimensionally invariant relation
\item
for any $(v_1,\ldots,v_n) \in V^n$ and any consistent lists $(s_1,\ldots,s_n)$ and $(t_1,\ldots,t_n)$ in $w_1 \times \ldots \times w_n$,
\[ (v_1s_1,\ldots,v_ns_n) \in S \: \Longleftrightarrow \: (v_1t_1,\ldots,v_nt_n) \in S \]
\item
for any $(v_1,\ldots,v_n) \in V^n$ and any consistent lists $(s_1,\ldots,s_n)$ and $(t_1,\ldots,t_n)$ in $w_1 \times \ldots \times w_n$,
\[ (v_1s_1,\ldots,v_ns_n) \in S \: \Longrightarrow \: (v_1t_1,\ldots,v_nt_n) \in S \]
\end{enumerate}
\hfill $\triangle$
\end{prop}

\noindent 2 and 4 are particularly important.

\begin{proof}
Easy; mainly just manipulations of definitions.
\end{proof}

\begin{note}
Dimensionally invariant sets and dimensionally invariant relations can be seen as different forms of exactly the same thing. Everything which we will say about one can immediately be turned into a statement about the other. We will work largely with dimensionally invariant sets, because it is simpler to deal with sets than to deal with indicator functions of sets.
\hfill $\triangle$
\end{note}

\begin{defi}
Let $w_1,\ldots,w_n$ be elements of $W$. Then we define $R[w_1,\ldots,w_n]$ as the set of dimensionally invariant subsets of $w_1 \times \ldots \times w_n$.

\textrm{} \\ We define $\overline{R}[w_1,\ldots,w_n] := \{ \overline{S} | S \in R[w_1,\ldots,w_n] \}$.
\hfill $\triangle$
\end{defi}

\begin{basicprop}
Let $w_1,\ldots,w_n$ be elements of $W$. Then the map
\begin{align*}
R[w_1,\ldots,w_n] & \to \overline{R}[w_1,\ldots,w_n] \\ 
                S & \mapsto \overline{S}
\end{align*}
is a bijection.
\hfill $\triangle$
\end{basicprop}

\begin{proof}
Easy.
\end{proof}

\begin{defi}
Let $w_1,\ldots,w_n$ be elements of $W$. Let $(s_1,\ldots,s_n)$ be a consistent list in $w_1 \times \ldots \times w_n$. Define the binary relation
\[ \sim_{(s_1,\ldots,s_n)} \]
on $w_1 \times \ldots \times w_n$ by the following: for all $(v_1,\ldots,v_n) \in V^n$ and $(x_1,\ldots,x_n) \in w_1 \times \ldots \times w_n$, \\

we have that
\[ (v_1s_1,\ldots,v_ns_n) \sim_{(s_1,\ldots,s_n)} (x_1,\ldots,x_n) \]

if and only if there exists a consistent list $(t_1,\ldots,t_n) \in w_1 \times \ldots \times w_n$ such that
\[ (x_1,\ldots,x_n) = (v_1t_1,\ldots,v_nt_n). \]
\hfill $\triangle$
\end{defi}

\noindent In other words, given any $(v_1s_1,\ldots,v_ns_n)$, the set of values $\mathbf{x}$ such that
\[ (v_1s_1,\ldots,v_ns_n) \sim_{(s_1,\ldots,s_n)} \mathbf{x} \]
is precisely the set of values that is obtained by replacing $(s_1,\ldots,s_n)$ with all other consistent lists in $w_1 \times \ldots \times w_n$.

\begin{defi}
Let $w_1,\ldots,w_n$ be elements of $W$. Define the binary relation $\sim$ on $w_1 \times \ldots \times w_n$ by the following: for all $(x_1,\ldots,x_n)$ and $(y_1,\ldots,y_n)$ in $w_1 \times \ldots \times w_n$, \\

we have that
\[ (x_1,\ldots,x_n) \sim (y_1,\ldots,y_n) \]

if and only if there exist $(v_1,\ldots,v_n) \in V^n$ and consistent lists $(s_1,\ldots,s_n)$ and $(t_1,\ldots,t_n)$ in $w_1 \times \ldots \times w_n$ such that
\begin{align*}
(v_1s_1,\ldots,v_ns_n) &= (x_1,\ldots,x_n) \\ 
(v_1t_1,\ldots,v_nt_n) &= (y_1,\ldots,y_n).
\end{align*}
In other words, for all $(x_1,\ldots,x_n)$ and $(y_1,\ldots,y_n)$ in $w_1 \times \ldots \times w_n$, we have that
\[ (x_1,\ldots,x_n) \sim (y_1,\ldots,y_n) \]
if and only if there exists a consistent list $(s_1,\ldots,s_n)$ in $w_1 \times \ldots \times w_n$ such that
\[ (x_1,\ldots,x_n) \sim_{(s_1,\ldots,s_n)} (y_1,\ldots,y_n). \]
\hfill $\triangle$
\end{defi}

\begin{note}
Take any $(v_1,\ldots,v_n) \in V^n$ and any consistent list $(s_1,\ldots,s_n) \in w_1 \times \ldots \times w_n$. If
\begin{align*}
(v_1s_1,\ldots,v_ns_n) & \sim_{(s_1,\ldots,s_n)} (x_1,\ldots,x_n) \textrm{ and} \\ 
(v_1s_1,\ldots,v_ns_n) & \sim_{(s_1,\ldots,s_n)} (y_1,\ldots,y_n)
\end{align*}
then clearly
\[ (x_1,\ldots,x_n) \sim (y_1,\ldots,y_n). \]
\hfill $\triangle$
\end{note}

\begin{lemma}
Let $w_1,\ldots,w_n$ be elements of $W$. For any consistent lists $(s_1,\ldots,s_n)$ and $(t_1,\ldots,t_n)$ in $w_1 \times \ldots \times w_n$
\[ \sim_{(s_1,\ldots,s_n)} = \sim_{(t_1,\ldots,t_n)}. \]
\hfill $\triangle$
\end{lemma}

\begin{proof}
The statement can be proved by showing that
\begin{enumerate}
\item
for any $(x_1,\ldots,x_n)$ and $(y_1,\ldots,y_n)$ in $w_1 \times \ldots \times w_n$,
\[ (x_1,\ldots,x_n) \sim_{(t_1,\ldots,t_n)} (y_1,\ldots,y_n) \: \Rightarrow \: (x_1,\ldots,x_n) \sim_{(s_1,\ldots,s_n)} (y_1,\ldots,y_n). \]
\item
for any $(x_1,\ldots,x_n)$ and $(y_1,\ldots,y_n)$ in $w_1 \times \ldots \times w_n$,
\[ (x_1,\ldots,x_n) \sim_{(s_1,\ldots,s_n)} (y_1,\ldots,y_n) \: \Rightarrow \: (x_1,\ldots,x_n) \sim_{(t_1,\ldots,t_n)} (y_1,\ldots,y_n). \]
\end{enumerate}
1. Any $(x_1,\ldots,x_n)$ and $(y_1,\ldots,y_n)$ that are related under $\sim_{(t_1,\ldots,t_n)}$ can be expressed as
\begin{align*}
(x_1,\ldots,x_n) &= (v_1t_1,\ldots,v_nt_n) \\
(y_1,\ldots,y_n) &= (v_1u_1,\ldots,v_nu_n)
\end{align*}

\noindent So take any $(v_1,\ldots,v_n) \in V^n$ and consistent list $(u_1,\ldots,u_n) \in w_1 \times \ldots \times w_n$; if we can show that
\[ (v_1t_1,\ldots,v_nt_n) \sim_{(s_1,\ldots,s_n)} (v_1u_1,\ldots,v_nu_n), \]
then we are done.

$\textrm{}$ \\ Let $c_1,\ldots,c_m$ be a basis of $\textrm{span}(t_1,\ldots,t_n)$.

$\textrm{}$ \\ $(u_1,\ldots,u_n)$ is consistent; and obviously, $q(c_1),\ldots,q(c_m)$ are linearly independent and their span is equal to $\textrm{span}(w_1,\ldots,w_n)$.
So let $(a_1,\ldots,a_m) \in V^n$ be such that
\[ p[q(c_1), \ldots, q(c_m); w_i](a_1c_1, \ldots, a_mc_m) = u_i \]
for each $i$ from $1$ to $n$.
So in particular,
\begin{align*}
u_i &= p[q(c_1), \ldots, q(c_m); w_i](a_1c_1, \ldots, a_mc_m) \\ 
    &= p[q(c_1), \ldots, q(c_m); w_i](a_1, \ldots, a_m) p[q(c_1), \ldots, q(c_m); w_i](c_1, \ldots, c_m) \\ 
    &= p[q(c_1), \ldots, q(c_m); w_i](a_1, \ldots, a_m) t_i
\end{align*}
for each $i$ from $1$ to $n$.
Now, for each $i$ from $1$ to $n$, let $z_i = v_i t_i s_i^{-1}$; so
\[ (v_1t_1,\ldots,v_nt_n) = (z_1s_1,\ldots,z_ns_n). \]
$(s_1,\ldots,s_n)$ is a consistent list in $w_1 \times \ldots \times w_n$; so let $d_1,\ldots,d_m$ be the elements respectively of $q(c_1),\ldots,q(c_m)$ such that
\[ p[q(c_1), \ldots, q(c_m); w_i](d_1, \ldots, d_m) = s_i \]
for each $i$ from $1$ to $n$.
Then, for each $i$ from $1$ to $n$,
\begin{align*}
v_iu_i &= v_i p[q(c_1),\ldots,q(c_m);w_i](a_1,\ldots,a_m)t_i \\ 
       &= z_i p[q(c_1),\ldots,q(c_m);w_i](a_1,\ldots,a_m) s_i \\ 
       &= z_i p[q(c_1),\ldots,q(c_m);w_i](a_1,\ldots,a_m) p[q(c_1), \ldots, q(c_m); w_i](d_1, \ldots, d_m) \\ 
       &= z_i p[q(c_1),\ldots,q(c_m);w_i](a_1d_1,\ldots,a_md_m)
\end{align*}
So if we let $l_i = p[q(c_1),\ldots,q(c_m);w_i](a_1d_1,\ldots,a_md_m)$ for each $i$ from $1$ to $n$, then:
just as
\[ (v_1t_1,\ldots,v_nt_n) = (z_1s_1,\ldots,z_ns_n), \]
so also
\[ (v_1u_1,\ldots,v_nu_n) = (z_1l_1,\ldots,z_nl_n). \]
$(l_1,\ldots,l_n)$ is consistent, because $l_1,\ldots,l_n \in \textrm{span}(a_1d_1,\ldots,a_md_m)$.

$\textrm{}$ \\ Hence $(v_1t_1,\ldots,v_nt_n) \sim_{(s_1,\ldots,s_n)} (v_1u_1,\ldots,v_nu_n)$.

$\textrm{}$ \\ 2. Just switch round the consistent lists $(s_1,\ldots,s_n)$ and $(t_1,\ldots,t_n)$ in 1.
\end{proof}

\noindent The above proof may appear complicated, but it is actually very simple in its essence: if one list of physical quantities can be ``rescaled'' to give another, when working in one particular list of units, then there is a well-defined concept of the scale-factors by which the rescaling was done, and these scale-factors still hold if we work in a different list of units.

$\textrm{}$ \\ This lemma immediately gives the following proposition.

\begin{prop}
Let $w_1,\ldots,w_n$ be elements of $W$. For any consistent list $(s_1,\ldots,s_n) \in w_1 \times \ldots \times w_n$,
\[ \sim_{(s_1,\ldots,s_n)} \ = \ \sim. \]
\hfill $\triangle$
\end{prop}

\begin{proof}
Obviously, for any $(x_1,\ldots,x_n)$ and $(y_1,\ldots,y_n)$ in $w_1 \times \ldots \times w_n$,
\[ (x_1,\ldots,x_n) \sim_{(s_1,\ldots,s_n)} (y_1,\ldots,y_n) \: \Rightarrow \: (x_1,\ldots,x_n) \sim (y_1,\ldots,y_n). \]
Now take any $(x_1,\ldots,x_n)$ and $(y_1,\ldots,y_n)$ in $w_1 \times \ldots \times w_n$ such that
\[ (x_1,\ldots,x_n) \sim (y_1,\ldots,y_n). \]
Then there exists a consistent list $(t_1,\ldots,t_n) \in w_1 \times \ldots \times w_n$ such that
\[ (x_1,\ldots,x_n) \sim_{(t_1,\ldots,t_n)} (y_1,\ldots,y_n). \]
Hence
\[ (x_1,\ldots,x_n) \sim_{(s_1,\ldots,s_n)} (y_1,\ldots,y_n). \]
\end{proof}

\begin{prop}
Let $w_1,\ldots,w_n$ be elements of $W$. Then $\sim$ is an equivalence relation.
\hfill $\triangle$
\end{prop}

\begin{proof}
$\sim$ is clearly reflexive and symmetric.

$\textrm{}$ \\ Suppose that for some $(x_1,\ldots,x_n)$, $(y_1,\ldots,y_n)$ and $(z_1,\ldots,z_n)$ in $w_1 \times \ldots \times w_n$,
\[ (x_1,\ldots,x_n) \sim (y_1,\ldots,y_n) \; \textrm{and} \; (y_1,\ldots,y_n) \sim (z_1,\ldots,z_n). \]
Then in particular, for some/any $(s_1,\ldots,s_n) \in w_1 \times \ldots \times w_n$,
\begin{align*}
(y_1,\ldots,y_n) & \sim_{(s_1,\ldots,s_n)} (x_1,\ldots,x_n) \; \textrm{and} \\ 
(y_1,\ldots,y_n) & \sim_{(s_1,\ldots,s_n)} (z_1,\ldots,z_n)
\end{align*}
and so
\[ (x_1,\ldots,x_n) \sim (z_1,\ldots,z_n). \]
\end{proof}

\begin{defi}
Let $w_1,\ldots,w_n$ be elements of $W$. Denote:
\begin{itemize}
\item
$Q[w_1,\ldots,w_n] := (w_1 \times \ldots \times w_n)/ \sim$
\item
for any $\mathbf{x} \in w_1 \times \ldots \times w_n$,
\[ Cl(\mathbf{x}):= \{ \mathbf{y} \in w_1 \times \ldots \times w_n | \mathbf{y} \sim \mathbf{x} \} \in Q[w_1,\ldots,w_n]. \]
\end{itemize}
\hfill $\triangle$
\end{defi}

\begin{prop}[dimensional invariance topology]
Let $w_1,\ldots,w_n$ be elements of $W$.
Then the set $R[w_1,\ldots,w_n]$ of dimensionally invariant subsets of $w_1 \times \ldots \times w_n$ is precisely the set of all unions of elements of $Q[w_1,\ldots,w_n]$.
In other words, for any subset $S$ of $w_1 \times \ldots \times w_n$, $S$ is dimensionally invariant if and only if there exists $\mathcal{U} \subset Q[w_1,\ldots,w_n]$ which forms a partition of $S$,
\[ S = \bigcup_{U \in \mathcal{U}} U \]
\end{prop}

\begin{proof}
It is easy to show that the union of any set of elements of $Q[w_1,\ldots,w_n]$ is dimensionally invariant.

$\textrm{}$ \\ To show that any dimensionally invariant subset of $S$ of $w_1 \times \ldots \times w_n$ is equal to a union of elements of $Q[w_1,\ldots,w_n]$, it is sufficient to show that
\[ \mathbf{x} \in S \: \Rightarrow \: Cl(\mathbf{x}) \subset S. \]
This too is easy.
\end{proof}

\begin{note}
$R[w_1,\ldots,w_n]$ is a topology on $w_1 \times \ldots \times w_n$, with a (unique) pairwise disjoint base $Q[w_1,\ldots,w_n]$. We shall call the elements of $Q[w_1,\ldots,w_n]$ \emph{basic dimensionally invariant subsets} of $w_1 \times \ldots \times w_n$.
\hfill $\triangle$
\end{note}

\begin{defi}
Let $w_1,\ldots,w_n$ be elements of $W$. Then define
\[ \overline{Q}[w_1,\ldots,w_n] := \{ \overline{B} | B \in Q[w_1,\ldots,w_n] \}. \]
\hfill $\triangle$
\end{defi}

\begin{prop}
Let $w_1,\ldots,w_n$ be elements of $W$. Then:
$\overline{R}[w_1\ldots,w_n]$ is a topology on $V^n$. For any consistent list $(s_1,\ldots,s_n) \in w_1 \times \ldots \times w_n$, let $\alpha_{(s_1,\ldots,s_n)}$ be the map
\begin{align*}
\alpha_{(s_1,\ldots,s_n)} : w_1 \times \ldots \times w_n & \to V^n \\ 
                                  (v_1s_1,\ldots,v_ns_n) & \mapsto (v_1,\ldots,v_n).
\end{align*}
Then
\[ \alpha_{(s_1,\ldots,s_n)}[S] = \overline{S} \]
for all $S \in R[w_1,\ldots,w_n]$, and
\[ \alpha_{(s_1,\ldots,s_n)} : (w_1 \times \ldots \times w_n, R[w_1,\ldots,w_n]) \to (V^n, \overline{R}[w_1\ldots,w_n]) \]
is a homeomorphism.
\hfill $\triangle$
\end{prop}

\begin{proof}
Take any consistent list $(s_1,\ldots,s_n) \in w_1 \times \ldots \times w_n$.
By definition, for any $S \in R[w_1,\ldots,w_n]$,
\[ \alpha_{(s_1,\ldots,s_n)}[S] = \overline{S}. \]
$\alpha_{(s_1,\ldots,s_n)}$ is clearly a bijection.
As established,
\[ \left. \left( \alpha_{(s_1,\ldots,s_n)}[.] \right) \right|_{R[w_1,\ldots,w_n]} \ : R[w_1,\ldots,w_n] \to \overline{R}[w_1\ldots,w_n] \]
is a bijection.
Everything else follows.
\end{proof}

\begin{cor}
Let $w_1,\ldots,w_n$ be elements of $W$. Then the elements of $\overline{Q}[w_1\ldots,w_n]$ form a partition of $V^n$, and any element of $\overline{R}[w_1\ldots,w_n]$ is equal to a union of elements of $\overline{Q}[w_1\ldots,w_n]$.
\hfill $\triangle$
\end{cor}

\noindent We now go on to explore the relationship between dimensionally invariant sets and nondimensionalising ``$\pi$ groups'' (FCLCFs, in our terminology).

\section*{3. Nondimensionalisation of the dimensional invariance topology}

\begin{defi}
Let $w_1,\ldots,w_n$ be elements of $W$. For any subset $S \subset w_1 \times \ldots \times w_n$, we define the indicator function
\[ I[w_1,\ldots,w_n]_S : w_1 \times \ldots \times w_n \to \{ \textrm{TRUE}, \textrm{FALSE} \} \]
by
\[ I[w_1,\ldots,w_n]_S(\mathbf{x}) = \textrm{TRUE} \iff \mathbf{x} \in S. \]
Let $k$ be a nonnegative integer. Then for any $A \subset V^k$, we define the indicator function
\[ I^{(k)}_A : V^k \to \{ \textrm{TRUE}, \textrm{FALSE} \} \]
by
\[ I^{(k)}_A(\mathbf{x}) = \textrm{TRUE} \iff \mathbf{x} \in S. \]
\hfill $\triangle$
\end{defi}

\begin{notat}
Let $A$ and $B$ be sets, and let $f_1,\ldots,f_n$ be functions from $A$ to $B$.
Let $S$ be a subset of $A$, and let $T$ be a subset of $B^n$ such that
\[ \{ (f_1(x),\ldots,f_n(x)) | x \in S \} \subset T. \]
Then we define the function
\[ (f_1,\ldots,f_n)_{(S,T)}: S \to T \]
by
\[ (f_1,\ldots,f_n)_{(S,T)}(x) = (f_1(x),\ldots,f_n(x)) \]
for all $x \in S$.
\hfill $\triangle$
\end{notat}

\noindent If we incorporate into the statement of the $\pi$-theorem some of the content of other notes and theorems which have appeared above, then we can be express the $\pi$-theorem as a statement about dimensionally invariant sets in the following way:

\begin{thm}[$\pi$-theorem reformulated]
Let $w_1,\ldots,w_n$ be elements of $W$, and let $r$ be such that the dimension of $\textrm{span}(w_1,\ldots,w_n)$ is $n-r$. Then there exists a basis $(\pi_1,\ldots,\pi_r)$ of $K[w_1,\ldots,w_n]$ such that: \\

for every dimensionally invariant subset $S \subset w_1 \times \ldots \times w_n$ there exists a set $A \subset V^r$ such that
\[ s = (\pi_1,\ldots,\pi_r)_{(w_1 \times \ldots \times w_n,V^r)}{}^{-1}[A]. \]
\hfill $\triangle$
\end{thm}

\noindent Observe that this statement, even though it actually contains \emph{more} information than the version given in section 1, is nonetheless more succinct than that version.

\begin{lemma}
Let $w_1,\ldots,w_n$ be elements of $W$, and let $r$ be such that the dimension of $\textrm{span}(w_1,\ldots,w_n)$ is $n-r$.
Let $(\pi_1,\ldots,\pi_r)$ and $(\psi_1,\ldots,\psi_r)$ be bases of $K[w_1,\ldots,w_n]$. Then there exist $r$-input FCLCFs $T_1,\ldots,T_r$ such that
\begin{itemize}
\item
for each $i$ from $1$ to $r$,
\[ \pi_i(.) = T_i( \psi_1(.), \ldots, \psi_r(.) ) \]
\item
the function
\[ (T_1,\ldots,T_r)_{(V^r,V^r)}: V^r \to V^r \]
is a bijection.\hfill $\triangle$
\end{itemize}
\end{lemma}

\begin{proof}
It is elementary linear algebra that each of $\pi_1,\ldots,\pi_r$ can be expressed as a linear combination of $(\psi_1,\ldots,\psi_r)$, and the $r \times r$ matrix of coefficients is invertible. Let $M$ be the matrix of coefficients; that is, we let the $i$-th row of $M$ be the list of coefficients in the expression of $\pi_i$ as a linear combination of $(\psi_1,\ldots,\psi_r)$.
Define $T_i$ as the FCLCF whose coefficients (in order) are the coefficients of $(\psi_1,\ldots,\psi_r)$ in the expression for $\pi_i$ --- i.e., define the coefficients of $T_i$ to be the values in the $i$-th row of $M$.
Clearly,
\[ \pi_i(.) = T_i( \psi_1(.), \ldots, \psi_r(.) ) \]
for each $i$ from $1$ to $r$.
Now take any $(v_1,\ldots,v_r) \in V^r$; then $(T_1,\ldots,T_r)_{(V^r,V^r)}(v_1,\ldots,v_r)$ is equal to
\[ ( \ v_1^{m_{1,1}} \ldots v_r^{m_{1,r}}, \ \ldots, \  v_1^{m_{r,1}} \ldots v_r^{m_{r,r}} \ ). \]
Diagrammatically, if we treat $(T_1,\ldots,T_r)_{(V^r,V^r)}(v_1,\ldots,v_r)$ as a column vector, then
\[ (T_1,\ldots,T_r)_{(V^r,V^r)}(v_1,\ldots,v_r) = M
\begin{pmatrix}
v_1 \\  \vdots \\  v_r
\end{pmatrix}
.\]
Let $N$ be the inverse of $M$. Then we can construct the inverse map $(T_1,\ldots,T_r)_{(V^r,V^r)}{}^{-1}$ defined on $V^r$ as
\begin{align*}
  & (T_1,\ldots,T_r)_{(V^r,V^r)}{}^{-1}(v_1,\ldots,v_r) \\ 
= & ( \ v_1^{n_{1,1}} \ldots v_r^{n_{1,r}}, \ \ldots, \  v_1^{n_{r,1}} \ldots v_r^{n_{r,r}} \ ).
\end{align*}
\end{proof}

\noindent We shall call the list $(T_1,\ldots,T_r)$ the \emph{transition from} $(\psi_1,\ldots,\psi_r)$ to $(\pi_1,\ldots,\pi_r)$.

\begin{defi}
Let $w_1,\ldots,w_n$ be elements of $W$, and let $r$ be such that the dimension of $\textrm{span}(w_1,\ldots,w_n)$ is $n-r$.
A basis of $K[w_1,\ldots,w_n]$ will be called ``special'' if (when expressed as an unordered set) it takes the form
\[ \left\{ (x_1,\ldots,x_n) \mapsto x_{l_i} p_i(x_{k_1},\ldots,x_{k_{n-r}}) | i \in \{1,\ldots,r\} \right\} \]
where
\[ \{1,\ldots,n\} = \{l_1,\ldots,l_r\} \cup \{k_1,\ldots,k_{n-r} \} \]
and $p_i$ is an $(n-r)$-input FCLCF for each $i$ from $1$ to $r$. 
\hfill $\triangle$
\end{defi}

\begin{note}
In the proof of the $\pi$-theorem, the basis $(\pi_1,\ldots,\pi_r)$ of $K[w_1,\ldots,w_n]$ which is constructed is a special basis.
\hfill $\triangle$
\end{note}

\begin{lemma}
Let $w_1,\ldots,w_n$ be elements of $W$, and let $(\psi_1,\ldots,\psi_r)$ be a special basis of $K[w_1,\ldots,w_n]$. Then the function
\[ (\psi_1,\ldots,\psi_r)_{(w_1 \times \ldots \times w_n,V^r)} : w_1 \times \ldots \times w_n \to V^r \]
is surjective.
\hfill $\triangle$
\end{lemma}

\begin{proof}
Express the basis $(\psi_1,\ldots,\psi_r)$ in the form
\[ \psi_i(x_1,\ldots,x_n) = x_{l_i} p_i(x_{k_1},\ldots,x_{k_{n-r}}) \ \forall (x_1,\ldots,x_n) \in w_1 \times \ldots \times w_n \]
for each $i$ from $1$ to $r$, where
\[ \{1,\ldots,n\} = \{l_1,\ldots,l_r\} \cup \{k_1,\ldots,k_{n-r} \} \]
and $p_i$ is an $(n-r)$-input FCLCF for each $i$ from $1$ to $r$.
Take any $(v_1,\ldots,v_n) \in V^r$. Fix a value $(c_1,\ldots,c_n) \in w_1 \times \ldots \times w_n$, and for each $i$ from $1$ to $r$, let
\[ u_i = \psi_i(c_1,\ldots,c_n). \]
Now, for each $i$ from $1$ to $r$, let
\[ x_{l_i} = v_i u_i^{-1} c_{l_i} \]
and for each $i$ from $1$ to $n-r$, let
\[ x_{k_i} = c_{k_i}. \]
Then for each $i$ from $1$ to $r$,
\begin{align*}
\psi_i(x_1,\ldots,x_n) &= v_i u_i^{-1} c_{l_i} p_i(c_{k_1},\ldots,x_{c_{n-r}}) \\
                       &= v_i u_i^{-1} \psi_i(c_1,\ldots,c_n) \\ 
                       &= v_i u_i^{-1} u_i \\ 
                       &= v_i
\end{align*}
So
\[ (\psi_1,\ldots,\psi_r)_{(w_1 \times \ldots \times w_n,V^r)}(x_1,\ldots,x_n) = (v_1,\ldots,v_n). \]
Thus $(\psi_1,\ldots,\psi_r)_{(w_1 \times \ldots \times w_n,V^r)}$ is surjective.
\end{proof}

\noindent From now on, for convenience and ease of reading, we shall abbreviate the subscript $_{(w_1 \times \ldots \times w_n,V^r)}$ to the subscript $_{\ast}$.

\begin{lemma}
Let $w_1,\ldots,w_n$ be elements of $W$. Then for any basis $(\pi_1,\ldots,\pi_r)$ of $K[w_1,\ldots,w_n]$, the function
\[ (\pi_1,\ldots,\pi_r)_{\ast} : w_1 \times \ldots \times w_n \to V^r \]
is surjective.
\hfill $\triangle$
\end{lemma}

\noindent In other words, every point in $V^r$ has a non-empty pre-image under $(\pi_1,\ldots,\pi_r)_{\ast}$.

\begin{proof}
Take any special basis $(\psi_1,\ldots,\psi_r)$ of $K[w_1,\ldots,w_n]$, and let $(T_1,\ldots,T_r)$ be the transition from $(\psi_1,\ldots,\psi_r)$ to $(\pi_1,\ldots,\pi_r)$. Then
\begin{align*}
  & (\pi_1,\ldots,\pi_r)_{\ast}[w_1 \times \ldots \times w_n] \\ 
= & ( \ T_1( \psi_1(.), \ldots, \psi_r(.) ), \ \ldots, \ T_r( \psi_1(.), \ldots, \psi_r(.) ) \ )_{\ast}[w_1 \times \ldots \times w_n] \\ 
= & (T_1,\ldots,T_r)_{(V^r,V^r)}[(\psi_1,\ldots,\psi_r)_{\ast}[w_1 \times \ldots \times w_n]] \\ 
= & (T_1,\ldots,T_r)_{(V^r,V^r)}[V^r] \\ 
= & V^r.
\end{align*}
\end{proof}

\begin{thm}[``Complete $\pi$-Theorem'', Version 1]
Let $w_1,\ldots,w_n$ be elements of $W$. For \emph{any} basis $(\pi_1,\ldots,\pi_r)$ of $K[w_1,\ldots,w_n]$,
\begin{enumerate}
\item
the dimensionally invariant subsets of $w_1 \times \ldots \times w_n$ are precisely the pre-images of the subsets of $V^r$ under $(\pi_1,\ldots,\pi_r)_{\ast}$
\item
the basic dimensionally invariant subsets of $w_1 \times \ldots \times w_n$ are precisely the pre-images of the singletons contained in $V^r$ under $(\pi_1,\ldots,\pi_r)_{\ast}$.
\end{enumerate}
\hfill $\triangle$
\end{thm}

\begin{proof}
Take any basis $(\pi_1,\ldots,\pi_r)$ of $K[w_1,\ldots,w_n]$.

$\textrm{}$ \\ 1. Take any subset $A \subset V^r$.
\begin{align*}
     & I[w_1,\ldots,w_n]_{(\pi_1,\ldots,\pi_r)_{\ast}{}^{-1}[A]}(x_1,\ldots,x_n) = \textrm{TRUE} \\ 
\iff & (x_1,\ldots,x_n) \in (\pi_1,\ldots,\pi_r)_{\ast}{}^{-1}[A] \\ 
\iff & (\pi_1,\ldots,\pi_r)_{\ast}(x_1,\ldots,x_n) \in A \\ 
\iff & I^{(r)}_A(\pi_1(x_1,\ldots,x_n),\ldots,\pi_r(x_1,\ldots,x_n)) = \textrm{TRUE}
\end{align*}
So $I[w_1,\ldots,w_n]_{(\pi_1,\ldots,\pi_r)_{\ast}{}^{-1}[A]} = I^{(r)}_A(\pi_1(.),\ldots,\pi_r(.))$.

$\textrm{}$ \\ $I^{(r)}_A(\pi_1(.),\ldots,\pi_r(.))$ is a dimensionally invariant relation (by one of the propositions in section 1). Thus $I[w_1,\ldots,w_n]_{(\pi_1,\ldots,\pi_r)_{\ast}{}^{-1}[A]}$ is a dimensionally invariant relation, and so $(\pi_1,\ldots,\pi_r)_{\ast}{}^{-1}[A]$ is a dimensionally invariant set.

$\textrm{}$ \\ Now, take any dimensionally invariant subset $S \in R[w_1,\ldots,w_n]$. Then $I[w_1,\ldots,w_n]_S$ is a dimensionally invariant relation. Hence, by the $\pi$-theorem, there exists an $r$-input dimensionless relation $J$ and a basis $(\psi_1,\ldots,\psi_r)$ of $K[w_1,\ldots,w_n]$ such that
\[ I[w_1,\ldots,w_n]_S = J(\psi_1(.),\ldots,\psi_r(.)). \]
Let
\[ A = \{ (v_1,\ldots,v_r) \in V^r | J(v_1,\ldots,v_r) = \textrm{TRUE} \}. \]
So $J = I^{(r)}_A$.
Then
\[ I[w_1,\ldots,w_n]_S = I^{(r)}_A(\psi_1(.),\ldots,\psi_r(.)). \]
Also, just as above,
\[ I[w_1,\ldots,w_n]_{(\psi_1,\ldots,\psi_r)_{\ast}{}^{-1}[A]} = I^{(r)}_A(\psi_1(.),\ldots,\psi_r(.)). \]
Hence
\[ I[w_1,\ldots,w_n]_S = I[w_1,\ldots,w_n]_{(\psi_1,\ldots,\psi_r)_{\ast}{}^{-1}[A]} \]
and so
\[ S = (\psi_1,\ldots,\psi_r)_{\ast}{}^{-1}[A]. \]

\noindent Now let $(T_1,\ldots,T_r)$ be the transition from $(\psi_1,\ldots,\psi_r)$ to $(\pi_1,\ldots,\pi_r)$. Let
\[ B=(T_1,\ldots,T_r)_{(V^r,V^r)}[A]. \]
$(T_1,\ldots,T_r)_{(V^r,V^r)}$ is bijective, and so
\begin{align*}
S &= (\psi_1,\ldots,\psi_r)_{\ast}{}^{-1}[A] \\ 
  &= (\psi_1,\ldots,\psi_r)_{\ast}{}^{-1}[(T_1,\ldots,T_r)_{(V^r,V^r)}{}^{-1}[B]] \\ 
  &= \left( (T_1,\ldots,T_r)_{(V^r,V^r)} \circ (\psi_1,\ldots,\psi_r)_{\ast} \right)^{-1}[B] \\ 
  &= (\pi_1,\ldots,\pi_r)_{\ast}{}^{-1}[B].
\end{align*}

$\textrm{}$ \\ 2. Follows immediately from part 1, using the facts that: the dimensionally invariant subsets of $w_1 \times \ldots \times w_n$ are precisely those which can be partitioned into basic dimensionally invariant subsets; and every point in $V^r$ has a non-empty pre-image under $(\pi_1,\ldots,\pi_r)_{\ast}$.
\end{proof}

\begin{cor}[``Complete $\pi$-Theorem'', Version 2]
Let $w_1,\ldots,w_n$ be elements of $W$. For any basis $(\pi_1,\ldots,\pi_r)$ of $K[w_1,\ldots,w_n]$,
\begin{enumerate}
\item
the elements of $\overline{R}[w_1,\ldots,w_n]$ are precisely the pre-images of the subsets of $V^r$ under $(\pi_1,\ldots,\pi_r)_{(V^n,V^r)}$
\item
the elements of $\overline{Q}[w_1,\ldots,w_n]$ are precisely the pre-images of the singletons contained in $V^r$ under $(\pi_1,\ldots,\pi_r)_{(V^n,V^r)}$.
\end{enumerate}
\hfill $\triangle$
\end{cor}

\begin{proof}
Take any consistent list $(s_1,\ldots,s_n) \in w_1 \times \ldots \times w_n$. Since
\[ \alpha_{(s_1,\ldots,s_n)} : (w_1 \times \ldots \times w_n, R[w_1,\ldots,w_n]) \to (V^n, \overline{R}[w_1\ldots,w_n]) \]
is a homeomorphism and we have the theorem above, it is sufficient just to show that
\[ (\pi_1,\ldots,\pi_r)_{\ast} \circ (\alpha_{(s_1,\ldots,s_n)}{}^{-1}) = (\pi_1,\ldots,\pi_r)_{(V^n,V^r)}. \]
So take any $(v_1,\ldots,v_n) \in V^n$. Then, for each $i$ from $1$ to $r$,
\begin{align*}
\pi_i \left( \alpha_{(s_1,\ldots,s_n)}{}^{-1}(v_1,\ldots,v_n) \right) &= \pi_i(v_1s_1,\ldots,v_ns_n) \\ 
&= \pi_i(v_1,\ldots,v_n) \pi_i(s_1,\ldots,s_n) \\ 
&= \pi_i(v_1,\ldots,v_n).
\end{align*}
So we are done.
\end{proof}

\newpage
\subsection*{References}

\textrm{[}1\textrm{]} Roberts, S.F., 1985. Applications of the theory of meaningfulness to psychology. \emph{Journal of Mathematical Psychology} \textbf{29}, pp.311--332.

$\textrm{}$ \\ \textrm{[}2\textrm{]} Sahal, D., 1976. On the conception and measurement of trade-off in engineering systems: A case study of the aircraft design process. \emph{Technological Forecasting and Social Change} \textbf{8}, pp.371--384.

$\textrm{}$ \\ \textrm{[}3\textrm{]} Barenblatt, G.I., 2003. \emph{Scaling}. Cambridge University Press, Cambridge.

$\textrm{}$ \\ \textrm{[}4\textrm{]} Curtis, W.D., Logan, J.D., Parker, W.A., 1982. Dimensional analysis and the pi theorem. \emph{Linear Algebra and its Applications} \textbf{47}, pp.117--126.

$\textrm{}$ \\ \textrm{[}5\textrm{]} Brand, L, 1957. The Pi theorem of dimensional analysis. \emph{Archive for Rational Mechanics and Analysis} \textbf{1}, pp.35--45.

\end{document}